\documentclass[letterpaper, 10 pt, conference]{ieeeconf}  %

\IEEEoverridecommandlockouts                              %
\usepackage{etoolbox}
\newtoggle{ext}
\usepackage{hyperref}
\toggletrue{ext}    

\usepackage[compact]{titlesec}
\titlespacing{\subsection}{0pt}{1ex}{0ex}
\renewcommand{\baselinestretch}{0.99}
\usepackage{graphics} %
\usepackage{amsmath} %
\usepackage{amssymb}  %
\usepackage{enumerate}
\usepackage{svg}
\usepackage[style=ieee,url=false,doi=false]{biblatex}

\usepackage{algpseudocode}
\usepackage{algorithm}

\usepackage{mathtools}
\DeclareMathAlphabet\mathbfcal{OMS}{cmsy}{b}{n}
\usepackage{bm}
\usepackage{amsthm}
\usepackage{amsmath}
\usepackage{xcolor}

\newtheorem{definition}{Definition}

\newtheorem{assumption}{Assumption}
\newtheorem{lemma}{Lemma}
\newtheorem{theorem}{Theorem}

\newtheorem{proposition}{Proposition}

\newtheorem{remark}{Remark}

\newtheorem{preremark3}{Theorem}[section]

\usepackage{mdframed}
\usepackage{lipsum}
\newmdtheoremenv{theo}{Theorem}

\providecommand{\R}{\ensuremath \mathbb{R}}

\providecommand{\N}{\ensuremath \mathbb{N}}

\newcommand{\ie}{\textit{i.e., }}
\newcommand{\eg}{\textit{e.g., }}

\newcommand{\norm}[1]{\left\Vert#1\right\Vert}

\newcommand{\defeq}{\vcentcolon=}

\DeclareMathOperator*{\minimize}{\mathrm{minimize}}

\newcommand{\eye}{{\mathbf{I}}}

\newcommand{\fstate}{\mathbf{x}}
\newcommand{\rstate}{\mathbf{x}_\mathrm{r}}
\newcommand{\nstate}{\mathbf{x}_\mathrm{n}}
\newcommand{\rdotstate}{\dot{\mathbf{x}}_\mathrm{r}}
\newcommand{\ndotstate}{\dot{\mathbf{x}}_\mathrm{n}}

\newcommand{\obs}{\mathbf{y}}

\newcommand{\tangent}{\mathbf{V}_\mathrm{r}}
\newcommand{\orthogonal}{\mathbf{V}_\mathrm{n}}

\newcommand{\xnerror}{\tilde{\fstate}_\mathrm{n}}
\newcommand{\xndoterror}{\dot{\tilde{\fstate}}_\mathrm{n}}

\newcommand{\zstate}{\mathbf{z}_\mathrm{r}}
\newcommand{\zdotstate}{\dot{\mathbf{z}}_\mathrm{r}}

\newcommand{\wmap}{\mathbf{w}} 
\newcommand{\vmap}{\mathbf{v}} 
\newcommand{\rom}{\mathbf{r}}

\newcommand{\An}{\mathbf{A}_\mathrm{n}}
\newcommand{\Ar}{\mathbf{A}_\mathrm{r}}
\newcommand{\B}{\mathbf{B}}

\newcommand{\C}{\mathbf{C}}
\newcommand{\Gj}{\mathbf{G}_\mathrm{j}}

\newcommand{\y}{\mathbf{y}}

\newcommand{\wnl}{\mathbf{w}_{\mathrm{nl}}} 
\newcommand{\rnl}{\mathbf{r}_{\mathrm{nl}}} 

\newcommand{\Br}{\mathbf{B}_\mathrm{r}} 
\newcommand{\Bn}{\mathbf{B}_\mathrm{n}} 

\newcommand{\manError}{\mathbf{e}}

\newcommand{\dd}{\mathbf{d}}

\newcommand{\ctrl}{\mathbf{u}}
\newcommand{\Cntrl}{\mathcal{U}}

\newcommand{\spectralsubspace}{E}

\newcommand{\fnl}{\mathbf{f}_{\mathrm{nl}}}

\newcommand{\nfstate}{{n_{\mathbf{f}}}}
\newcommand{\nrstate}{n}
\newcommand{\nctrl}{m}

\usepackage{comment}

\title{\LARGE \bf
Robust Nonlinear Reduced-Order Model Predictive Control
}

\author{John Irvin Alora$^{\star, 1}$, Luis A. Pabon$^{\star, 1}$, Johannes K\"ohler$^2$, Mattia Cenedese$^3$, \\
 Ed Schmerling$^1$, Melanie N. Zeilinger$^2$, George Haller$^3$, Marco Pavone$^1$%
 \thanks{$^\star$The first two authors contributed equally to this work}
\thanks{$^1$Department of Aeronautics and Astronautics, Stanford University (e-mail: \{jjalora, lpabon, schmrlng, pavone\}@stanford.edu).}%
\thanks{$^2$Institute for Dynamic Systems and Control, ETH Zürich, Zürich CH-8092, Switzerland (e-mail: \{jkoehle, mzeilinger\}@ethz.ch.}%
\thanks{$^3$Institute for Mechanical Systems, ETH Z\"{u}rich (e-mail: \{mattiac, georgehaller\}@ethz.ch).}%
\thanks{J.A. is supported by the Secretary of the Air Force STEM Ph.D. Fellowship. This work was supported by the NASA University Leadership Initiative (grant \#80NSSC20M0163) and KACST; this article solely reflects the opinions and conclusions of its authors and not any Air Force, NASA, nor KACST entity.}%
\thanks{J.K. was supported by an ETH Career Seed Award funded through the ETH Zurich Foundation and Swiss National Science Foundation under NCCR Automation (grant agreement 51NF40 180545).}%
}

\begin{document}
\maketitle
\thispagestyle{empty}
\pagestyle{empty}

\begin{abstract}
Real-world systems are often characterized by high-dimensional nonlinear dynamics, making them challenging to control in real time. While reduced-order models (ROMs) are frequently employed in model-based control schemes, dimensionality reduction introduces model uncertainty which can potentially compromise the stability and safety of the original high-dimensional system. In this work, we propose a novel reduced-order model predictive control (ROMPC) scheme to solve constrained optimal control problems for nonlinear, high-dimensional systems.
To address the challenges of using ROMs in predictive control schemes, we derive an error bounding system that dynamically accounts for model reduction error. Using these bounds, we design a robust MPC scheme that ensures robust constraint satisfaction, recursive feasibility, and asymptotic stability. We demonstrate the effectiveness of our proposed method in simulations on a high-dimensional soft robot with nearly 10,000 states. 
\end{abstract}

\section{INTRODUCTION}
High-dimensional dynamical systems, \eg derived from continuum mechanics, arise in various fields of science and engineering, including robotics, aerospace, chemical engineering, and neuroscience. In many of these applications, ensuring the safe operation of these systems is of utmost importance. Unfortunately, the high dimensionality of these models poses significant computational challenges when used in online optimal control schemes that enforce safety constraints, such as model predictive control (MPC) \cite{kouvaritakis2016model}. 
Model reduction is an effective approach to mitigate this computational bottleneck \cite{antoulas2000survey} and enable real-time control. The problem is that model uncertainty stemming from dimensionality reduction can cause the controlled high-dimensional system to violate critical constraints, ultimately compromising stability and safety.
\begin{figure}
\centering\includegraphics[width=0.5\textwidth]{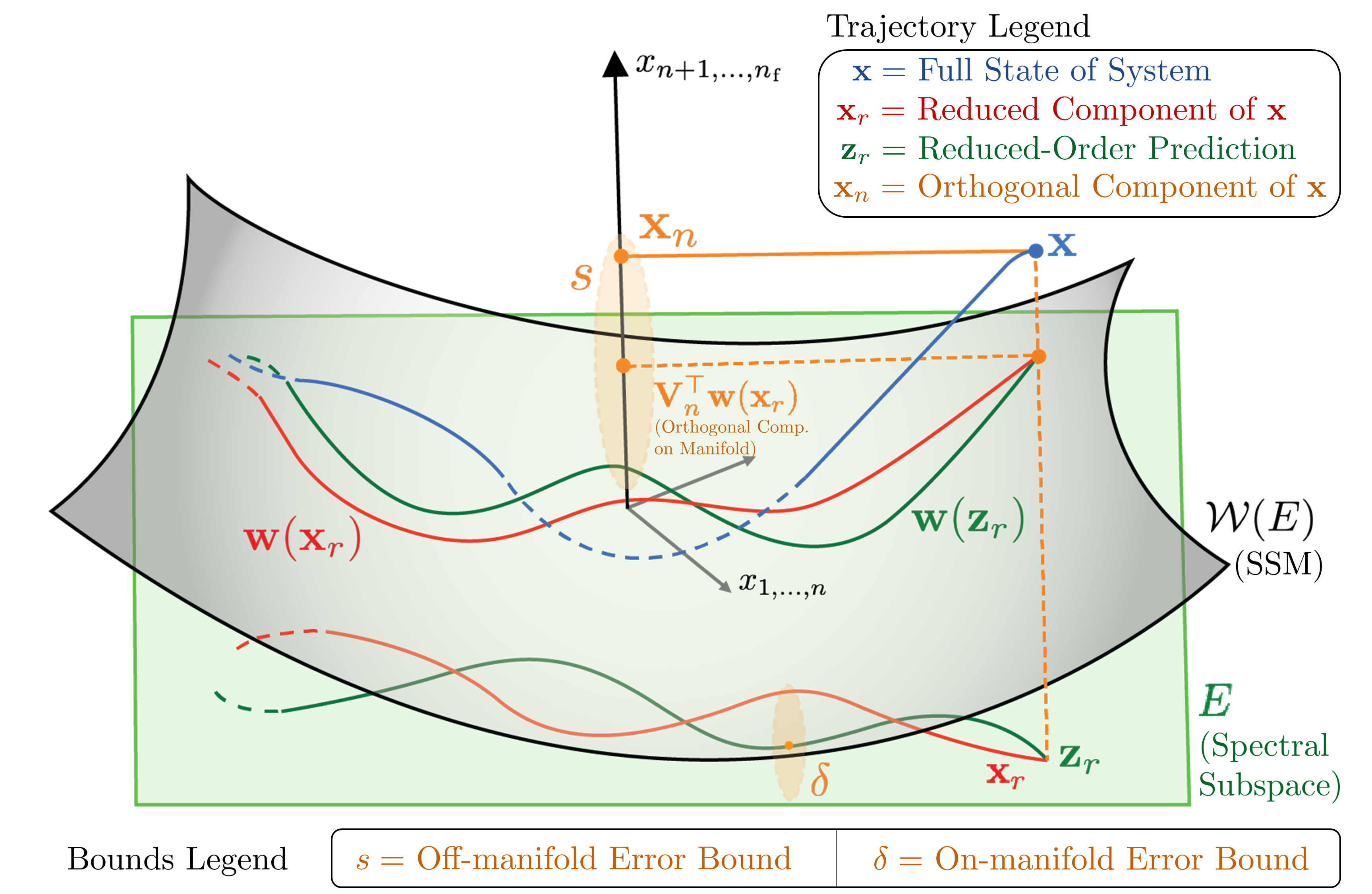}%
    \caption{We characterize the model reduction error of ROMs on invariant manifolds by decomposing the error into off-manifold ($\fstate_n - \orthogonal^\top\wmap(\rstate)$) and on-manifold components ($\rstate - \zstate$). We derive scalar error dynamics $s$ and $\delta$, which upper-bound the off-manifold and on-manifold error, respectively.}
    \label{fig:manifoldTube}
    \vspace{-0.65cm}
\end{figure}

\textbf{Statement of Contributions}:
In this work, we propose a new method for real-time control of high-dimensional, nonlinear systems that robustly satisfies constraints. 
We leverage recent advancements in model reduction, namely \textit{Spectral Submanifolds} (SSMs) to extract low-dimensional models suitable for real-time control. Our contributions are three-fold: 

(i) We quantify the modelling error due to SSM-based model reduction as depicted in Figure~\ref{fig:manifoldTube}. We derive an error bounding dynamical system of the model reduction error for reduced-order models (ROMs) that evolve on an invariant manifold of the autonomous system. 

(ii) We leverage these error bounds to design a robust, nonlinear reduced order model predictive control (RN-ROMPC) scheme that guarantees stability and robust constraint satisfaction. 

(iii) Lastly, we validate our approach via simulation on a $9768$-dimensional soft robot finite element model.

\textbf{Related Work}: 
Nonlinear model reduction provides a rigorous framework for constructing low-dimensional surrogate models for control. While these techniques have been applied successfully in various practical applications \cite{thieffry_control_2019, TonkensLorenzettiEtAl2021}, it is generally difficult to guarantee that the control scheme will be robust to model reduction error for generic classes of dynamical systems. Although several efforts have derived error bounds for ROMs \cite{perev_balanced_2012, koc_optimal_2021, buchfink_symplectic_2021}, they are often restricted to a limited range of applicable systems or require specific system structures that are difficult to verify. Furthermore, none of these approaches leverage their error bounds to achieve robust performance under constraints. 

Motivated by the successful application of SSMs to nonlinear model reduction and control \cite{AloraCenedeseEtAl2023}, we derive prediction error bounds for a general class of nonlinear systems. Specifically, we leverage the invariance property of SSMs to construct stable error bounds and design a control scheme that guarantees robust constraint satisfaction under model reduction error. While robust \textit{reduced order model predictive control} (ROMPC) schemes have been established for linear systems \cite{sopasakis2013, lohning_model_2014,lorenzetti2022linear}, we extend this line of work to develop a robust nonlinear ROMPC scheme. While there exist nonlinear robust MPC schemes \cite{falugi2013getting, sasfi2022robust,rakovic2022homothetic}, to the best of our knowledge, our effort constitutes the first line of work towards designing robust MPC schemes for nonlinear ROMs.

\section{PRELIMINARIES}
\label{sec:prelims}
\textbf{Notation}:
We denote $\frac{\partial \mathbf{f}}{\partial \mathbf{x}} \defeq \mathbf{f}'(\mathbf{x})$ as the Jacobian of a function $\mathbf{f}$ with respect to $\fstate$. The vector 2-norm and its induced matrix norm are denoted by $\norm{\cdot}$ while $\operatorname{Re}(z)$ denotes the real part of a complex number $z$.
For a Lipschitz continuous function $\mathbf{f}:\mathbb{R}^n\rightarrow\mathbb{R}^m$, $L_{\mathbf{f}}>0$ denotes its Lipschitz constant, i.e., $\norm{\mathbf{f}(x)-\mathbf{f}(z)}\leq L_{\mathbf{f}}\norm{x-z}$. Lastly, $\operatorname{alg}(\lambda)$ and $\operatorname{geom}(\lambda)$ denote the algebraic and geometric multiplicity of an eigenvalue $\lambda$, respectively. For brevity, we refer the reader to proofs of the lemmas and Proposition~\ref{prop:tighten} in the Appendix
\iftoggle{ext}{\ref{app:Appendix}.
}{of the extended version: {\small \url{https://tinyurl.com/2677ac9y}}}

\subsection{System Dynamics}
We consider the following high-dimensional nonlinear system with an equilibrium point at the origin
\begin{equation}
\begin{aligned}
\label{eq:FOM}
    \dot{\mathbf{\fstate}}(t) = \mathbf{A} \fstate (t) + \fnl(\fstate (t)) + \mathbf{B} \ctrl (t) + \dd(t),\\
\end{aligned}
\end{equation}
where $\nfstate \gg 1$ and $m$ are the dimensions of the full state and the input, respectively, $\mathbf{A} \in \R^{\nfstate \times \nfstate}$ and $\fnl : \R^\nfstate \rightarrow \R^\nfstate$ represent the linear and nonlinear parts of the dynamics, respectively, while $\mathbf{B} \in \R^{\nfstate \times m}$ represents the linear control matrix. The disturbance term $\dd(t)$ is assumed bounded \ie $\norm{\dd(t)} \leq \bar{d}$ for all $t \geq 0$.
We require that $\mathbf{A}$ and $\fnl$ satisfy the following assumptions.
\begin{assumption}[Asymptotic Stability and Semi-Simplicity]
\label{assum:stable}
$\mathbf{A}$ is a Hurwitz matrix, \ie each eigenvalue $\lambda_i$ of $\mathbf{A}$ has $\operatorname{Re}(\lambda_i) < 0$. Also, $\mathbf{A}$ is semi-simple \ie $\operatorname{alg}(\lambda_i) = \operatorname{geom}(\lambda_i)$.
\end{assumption}

\begin{assumption}[Analytic Nonlinearities]
\label{assum:smooth}
The nonlinear term, $\fnl \in \mathcal{C}^\infty$ satisfies $\fnl(\mathbf{0}) = \mathbf{0}, \, \fnl'(\mathbf{0}) = \mathbf{0}$ and
is $L_{\fnl}$-Lipschitz.
\end{assumption}

Assumption~\ref{assum:stable} requires that the origin is open-loop stable\footnote{More generally, this can be relaxed to stabilizability where a controller can be designed to stabilize the origin.} while semi-simplicity implies that $\mathbf{A}$ can be uniquely decomposed into a set of real, unique eigenspaces (cf.~\cite{haller2016nonlinear, ponsioen2020model}). 
Assumption~\ref{assum:smooth} requires that the system exhibit smooth behavior\footnote{Recent work \cite{Bettini2023nonsmooth} relaxes this assumption for mild discontinuities such as those due to dry friction, etc.}. 
These assumptions generally hold for many physical dissipative systems of interest, including soft robots, fluid flow, and chemical reactions.

\subsection{Problem Statement}
In this work, we consider the control of system \eqref{eq:FOM} subject to constraints on its inputs, $\ctrl$, and performance variables, $\y = \C \fstate \in \R^{n_{\y}}$, of the form 
$$\y(t) \in \mathcal{Y},\; \ctrl(t) \in \mathcal{U}, \quad t \geq 0,$$
where $\mathcal{Y}$ and $\mathcal{U}$ are compact sets. The performance constraint set $\mathcal{Y}$ is defined as 
\begin{equation}
    \mathcal{Y} :=\left\{ \y \in \R^{n_{\y}} \mid h_j(\y) \leq 0, \; j = 1, \ldots, n_{h}\right\} \label{eq:constraints}
\end{equation}
where $n_h \in \N$ represents the number of scalar constraints and each $h_j$ is $L_{h_j}$-Lipschitz.
We consider the problem of controlling System~\eqref{eq:FOM} to track dynamic trajectories $(\bar{\y}(t), \bar{\ctrl}(t))$ while satisfying the aforementioned constraints:
\begin{align*}
        \min_{\ctrl(\cdot), \fstate(\cdot)}&~ \int_{0}^{\infty} \ell(\y(\tau), \ctrl(\tau))\mathrm{d}\tau \\
    \mathrm{s.t.}~
        & \text{System}~\eqref{eq:FOM}, \\
        & \y = \C \fstate, \\
        & \y(t) \in \mathcal{Y}, \, \ctrl(t) \in \mathcal{U}.
\end{align*}
where $\ell$ is a positive definite stage cost with respect to $\y$ and $\ctrl$.

Unfortunately, the resulting optimal control problem (OCP) is computationally intractable since $\nfstate \gg 1$. We apply model reduction using SSMs to approximate System~\eqref{eq:FOM} with a low-dimensional surrogate model, then reason about the resulting model reduction error to design a computationally-tractable robust MPC scheme that ensures the high-dimensional system satisfies constraints in closed-loop. In the following, we summarize results on the existence of SSMs for System~\eqref{eq:FOM} and define some of its properties.

\subsection{SSM Basics}
Consider the autonomous part of System~\eqref{eq:FOM}, \ie $\dd\equiv 0$, $\ctrl\equiv 0$,
\begin{align}
\label{eq:FOM_aut}
    \dot{\mathbf{\fstate}}(t) = \mathbf{A} \fstate (t) + \fnl(\fstate (t)),
\end{align}
where each eigenvalue $\lambda_j$ of $\mathbf{A}$ corresponds to an eigenspace $E_j\subset \R^{\nfstate}$ spanned by its associated (generalized) eigenvectors. These eigenspaces are invariant subspaces for the linearized system. 

Since $\mathbf{A}$ is semi-simple, it is diagonalizable, and we may decompose it in block-diagonal form with real eigenvectors \cite{antoulas2005approximation}. 
Without loss of generality and for ease of exposition, we make the following assumption:
\begin{assumption}[Modal Coordinates]
    \label{assum:modalCoords}
    System~\eqref{eq:FOM} is in modal coordinates such that $\mathbf{A}$ is in real block-diagonal form whose blocks are ordered from slowest to fastest modes.
\end{assumption}
We emphasize that Assumption~\ref{assum:modalCoords} is made to simplify the exposition and, in practice, we do not need to diagonalize the full system~\eqref{eq:FOM} as we will discuss in Section~\ref{subsec:dataDriven}.

We now define the matrix $\mathbf{V} = \begin{bmatrix} \tangent, \orthogonal \end{bmatrix}\in\mathbb{R}^{\nfstate\times \nfstate}$ where $\tangent = \begin{bmatrix} \mathbf{I}_{\nrstate \times \nrstate}^\top, \mathbf{0}^\top_{\nfstate-\nrstate \times \nrstate} \end{bmatrix}^\top$ and 
$\orthogonal = \begin{bmatrix} \mathbf{0}_{\nrstate \times \nrstate}^\top, \mathbf{I}^\top_{\nfstate-\nrstate \times \nrstate} \end{bmatrix}^\top$. 
Under Assumption~\ref{assum:modalCoords}, the columns of $\tangent$ represent the eigenspace spanned by the $\nrstate$ slowest modes, which we denote as the \textit{spectral subspace}, $\spectralsubspace$, while the columns of $\orthogonal$ represent its complement. The spectrum of $\spectralsubspace$ is denoted $\Lambda_{\spectralsubspace}$, while the outer remaining eigenvalues of $\mathbf{A}$ are collected in the spectrum $\Lambda_{\mathrm{out}}$.
By Assumption~\ref{assum:modalCoords}, the linear matrix takes the form $\mathbf{A} = \mathrm{diag}(\Ar,\An)$, where $\Ar=\tangent^\top \mathbf{A}\tangent$ and $\An=\orthogonal^\top \mathbf{A}\orthogonal$. Thus, we always have that
\begin{align}
    \tangent^\top \tangent = \eye, \; \orthogonal^\top \tangent = \mathbf{0}, \tangent \tangent^\top + \orthogonal \orthogonal^\top = \eye, \label{eq:orthbasis}\\
    \tangent^\top \mathbf{A} \orthogonal = \mathbf{0}, \orthogonal^\top \mathbf{A} \tangent = \mathbf{0}. \label{eq:Ar}
\end{align}

Recent developments in nonlinear dynamics have shown the existence and uniqueness of smooth invariant structures for system \eqref{eq:FOM_aut}. These structures, known as SSMs, are the nonlinear extensions of the spectral subspaces of the linearization of \eqref{eq:FOM_aut}. The SSM corresponding to a spectral subspace $E$ can be defined as follows.

\begin{definition}
    \label{aut_SSM}
    An autonomous SSM $\mathcal{W}(E)$ corresponding to a spectral subspace $E$ is an invariant manifold of the autonomous part \eqref{eq:FOM_aut} of the nonlinear system \eqref{eq:FOM}, \ie when $\ctrl\equiv 0$, $\dd\equiv 0$,
    $$
    \mathbf{x}(0) \in \mathcal{W}(E) \Longrightarrow \mathbf{x}(t) \in \mathcal{W}(E), \quad \forall t \in \mathbb{R},
    $$
    such that,
    \begin{enumerate}
        \item $\mathcal{W}(E)$ is tangent to $E$ at the origin and has the same dimension as $E$,
        \item $\mathcal{W}(E)$ is strictly smoother than any other invariant manifold satisfying condition 1 above.
    \end{enumerate}
\end{definition}

SSMs as described in Definition~\ref{aut_SSM} are guaranteed to exist and to be unique as long as an additional non-resonance assumption holds.
\begin{assumption}[Non-Resonance Condition]
\label{assum:nonresonance}
The spectrum, $\Lambda_{\spectralsubspace}$, has no eigenvalue that is an integer combination of any eigenvalues in the outer spectrum $\Lambda_{\text{out}}$ (see \cite{haller2016nonlinear, cenedese2022ssmlearn} for details).
\end{assumption}
This assumption ensures that the nonlinear interactions between the slow modes in $\Lambda_{\spectralsubspace}$ and fast modes in $\Lambda_{\text{out}}$ are weak and that the ROM captures all strongly interacting modes. In general, it is generically satisfied, and we could also enlarge $E$ to contain all resonant modes of $\mathbf{A}$.

SSMs are effective for model reduction (for $\dd\equiv 0$, $\ctrl\equiv 0$) because trajectories of the full system are exponentially attracted to the manifold and synchronize with the slow dynamics evolving on it \cite{haller2016nonlinear}.

\subsection{Reduced Order Model}
We now construct a ROM on the SSM, $\mathcal{W}(E)$, corresponding to the $n$-slowest decaying modes.\footnote{We assume our constraint set $\mathcal{Y}$ is chosen non-restrictive enough such that $\left\{\C \fstate \mid \fstate \in \mathcal{W}(E)\right\} \cap \mathcal{Y} \neq \emptyset$.} We parameterize $\mathcal{W}(E)$ as a graph tangent to its spectral subspace, $E$, at the origin. Following the graph-style approach of \cite{cenedese2022ssmlearn}, our SSM parametrization is
\begin{equation}
    \begin{aligned}
    \label{eq:rom}
        \rstate(t) &= \vmap(\fstate(t)) \defeq \tangent^\top \fstate(t), \\
        \fstate(t) &= \wmap(\rstate(t)) \defeq \tangent \rstate(t) + \wnl(\rstate(t)), \\
    \end{aligned}
\end{equation}
where the mapping $\vmap$ projects the full state $\fstate$ onto the reduced coordinates in $E$, while the parameterization $\wmap$ maps the reduced state $\rstate$ onto the SSM $\mathcal{W}(E)$ in the full state space. 

By Definition~\ref{aut_SSM}, the graph parametrization \eqref{eq:rom} satisfies invertibility 
    \begin{equation}\label{eq:invert} 
        \begin{aligned} 
            \fstate = \wmap(\vmap(\fstate)) \text{ and } \rstate = \vmap(\wmap(\rstate)),
    \end{aligned} \end{equation}
and invariance, as stated in Definition~\ref{aut_SSM},
\begin{equation}\label{eq:invar} 
    \begin{aligned} 
        \mathbf{A} \wmap(\rstate) + \fnl(\wmap(\rstate)) =  \wmap'(\rstate) \rdotstate,
\end{aligned} \end{equation}
where $\rdotstate$ is evaluated for the autonomous dynamics ($\dd\equiv 0$, $\ctrl\equiv 0$) of the reduced system.

Using this, we now construct the reduced-order autonomous dynamics on $\mathcal{W}(E)$. These reduced dynamics approximate the behavior of the autonomous system~\eqref{eq:FOM_aut} using the slowest modes $\rstate$.
\begin{lemma}
\label{lem:redAutDyn}
    The reduced-order autonomous dynamics of System~\eqref{eq:FOM_aut} on the SSM, $\mathcal{W}(E)$, are
    \begin{equation}
    \begin{aligned}
    \label{eq:rautdyn}
        \rdotstate(t) &= \rom(\rstate(t)) \defeq \Ar \rstate (t) + \rnl(\rstate (t)).
    \end{aligned}
\end{equation}
\end{lemma}
\begin{remark}
The fast modes denoted by $\nstate$ on the manifold reduce to $ \nstate = \orthogonal^\top \wmap(\rstate) \stackrel{\eqref{eq:orthbasis}}{=} \orthogonal^\top \wnl(\rstate)$.
\end{remark}

We leverage the fact that System~\eqref{eq:FOM} is smooth (see~Assumption~\ref{assum:smooth}) to pick an appropriate functional form for $\wnl$ and $\rnl$. In this work, we construct these mappings by Taylor series expansion which naturally leads to the following assumption on the form of $\wnl$.
\begin{assumption}[Smoothness of Parameterization]
    \label{assum:contdiff}
    The nonlinear term in the parameterization, $\wnl$ is continuously differentiable and $L_{\wnl}$-Lipschitz.
\end{assumption}
Note that from Assumptions \ref{assum:smooth} and \ref{assum:contdiff} we get that $\rnl$ is $L_{\rnl}$-Lipschitz, with $L_{\rnl} \leq L_{\fnl} \left( 1  + L_{\wnl}\right)$. Equipped with these properties, we can now derive rigorous prediction error bounds for the reduced-order dynamics.
\section{PREDICTION ERROR BOUNDS}
\label{sec:errorBounds}
In this section, we derive prediction error bounds for the ROM. Specifically, we introduce the effect of input and disturbance and decompose the error dynamics into off-manifold and on-manifold components as shown in Figure~\ref{fig:manifoldTube}.
We then use this decomposition to construct scalar error dynamics, which give bounds on the prediction error of the ROM.

In the following lemmas, we derive properties of the SSM parameterization \eqref{eq:rom} and reduced dynamics \eqref{eq:rautdyn} that will be useful to construct the form of the scalar error dynamics.

\begin{lemma}
    \label{lem:SSMproperties}
    For all $\rstate \in \R^{\nrstate}$, it holds that:
        \begin{subequations}
            \begin{align}
                \label{eq:VnVnWnl}
                &\orthogonal \orthogonal^\top \wnl(\rstate) = \wnl(\rstate),\\
                \label{eq:VnDWnl}
                &\orthogonal^\top \wnl'(\rstate) \rom(\rstate) = \orthogonal^\top \left( \mathbf{A} \wnl(\rstate) + \fnl(\wmap(\rstate)) \right).
            \end{align}
        \end{subequations}
\end{lemma}

Lemma~\ref{lem:SSMproperties} allows us to derive the dynamics of the true system \eqref{eq:FOM} in the modal coordinates defined by $\mathbf{V} = \begin{bmatrix} \tangent, \orthogonal \end{bmatrix}$. This is done in the following lemma, where we decompose the dynamics into its slow and fast components.

\begin{lemma}
    \label{lem:trueDynModal}
    In modal coordinates, we can represent the controlled true system dynamics \eqref{eq:FOM} as follows
    \begin{equation} \label{eq:truemodal}
        \begin{gathered}
            \rdotstate = \rom(\rstate) + \tangent^\top \left(\manError(\fstate) + \B \ctrl +  \dd \right) \\
            \ndotstate = \An \nstate + \orthogonal^\top \left( \fnl(\fstate) + \B \ctrl +  \dd \right)
        \end{gathered}
    \end{equation}
    where $\manError(\fstate) \defeq \fnl(\fstate) - \fnl(\wmap(\rstate))$.
\end{lemma}

Lemma~\ref{lem:trueDynModal} puts the dynamics of System~\eqref{eq:FOM} in a convenient form for analysis as it reveals how the error and disturbances contribute to dynamics on and off the manifold. For example, notice that there is a form of residual error, $\manError(\fstate)$, in the reduced dynamics. If the true system remains on the manifold for all time \ie $\fstate = \wmap(\rstate)$, then $\manError(\fstate)$ is exactly zero. If the true system is ever off the manifold (as shown in Figure~\ref{fig:manifoldTube}), then the effect of the faster modes, $\nstate$, results in a disturbance in the reduced coordinates. 
In this case, the dynamics of the true system are no longer synchronized with the reduced dynamics on the manifold. Additionally, the effect of control acts as a disturbance that, in some directions, pushes the true system off the manifold. Figure~\ref{fig:manifoldTube} depicts this interplay between these reduced and orthogonal dynamics. 

Using this insight, we now introduce a linear input to Equation~\eqref{eq:rautdyn} and seek to characterize the error dynamics of the controlled prediction model. 
We define the nominal dynamics by ignoring the disturbance $\dd$ and assuming that $\fstate=\wmap(\rstate)$, resulting in the prediction model
\begin{equation} \label{eq:nominal}
    \begin{aligned}
        \zdotstate \defeq& \Ar \zstate + \rnl(\zstate) + \Br \ctrl, \\
    \end{aligned}
\end{equation}
where $\zstate$ is the prediction of $\rstate$ used in MPC, while $\rstate$ denotes the (unknown) true reduced state, as shown in Figure~\ref{fig:manifoldTube}. The matrix $\Br \defeq \tangent^\top \mathbf{B}$ is the projection of the linear control matrix onto the reduced coordinates while $\Bn \defeq \orthogonal^\top \mathbf{B}$ is the projection of the control onto the orthogonal space.

To construct tubes around the nominal dynamics, we must upper bound the error $\manError(\fstate)$. 

\begin{lemma}
    \label{lem:manError}
    The residual term in the orthogonal direction $\manError(\fstate)$ defined in \eqref{eq:truemodal} is upper bounded by the distance between the orthogonal component of the full state and its corresponding orthogonal component on the manifold,
    \begin{equation}
        \label{eq:manError}
        \norm{\manError(\fstate)} \leq L_{\fnl} \norm{\nstate - \orthogonal^\top \wnl(\rstate)},
    \end{equation}
    where $L_{\fnl}$ is the Lipschitz constant of $\fnl(\fstate)$ introduced in assumption \ref{assum:smooth}.
\end{lemma}

The off-manifold error dynamics between our control prediction model \eqref{eq:nominal} and the true system \eqref{eq:FOM} is given by the following lemma.
\begin{lemma}
\label{lem:orthError}
    Denote the difference between the orthogonal component of the full state $\fstate$ and its corresponding fast state on the manifold as
    \begin{align}
    \label{eq:x_n_tilde}
    \xnerror \defeq \nstate - \orthogonal^\top \wnl(\rstate).     
    \end{align}
    The off-manifold error dynamics then take the form
    \begin{equation} \label{eq:xndoterror}
        \xndoterror = \An \xnerror + \orthogonal^\top \left(\eye  - \wnl'(\rstate) \tangent^\top \right) \left( \mathbf{B} \ctrl + \manError(\fstate) + \dd \right).
    \end{equation}
\end{lemma}

Notice that the combined effect of error, input, and disturbance, \ie $\mathbf{B} \ctrl + \manError(\fstate) + \dd$, dictate how much the true system deviates from the manifold. In the linear case ($\wnl'(\rstate)=0$) the error term only affects the system in the direction $\orthogonal^\top$, orthogonal to the spectral subspace. The nonlinear case is similar but with an added term ($- \wnl'(\rstate) \tangent^\top$) that accounts for the curvature of the manifold. 

We can now derive scalar error dynamics, which bound both off and on-manifold errors. This construction allows us to dynamically change the size of the uncertainty tube around our predictions commensurate with the magnitudes of the input, error, and disturbance.

\begin{proposition}
\label{prop:tubedyn}
    Consider any initial condition $\fstate(0)\in\mathbb{R}^{n_f}$, $\zstate(0)\in\mathbb{R}^{\nrstate}$, $s(0),\delta(0)\geq 0$, such that $\norm{\xnerror(0)} \leq s(0)$ and $\norm{\rstate(0) - \zstate(0)} \leq \delta(0)$.
    Then, for any input signal $u(t)$ and any disturbance $\dd(t)$, $\|\dd(t)\|\leq \overline{d}$, it holds that
    \begin{subequations} \label{eq:tube_bounds}
        \begin{align}
            \norm{\xnerror(t)} &\leq s(t), \quad t \geq 0, \label{eq:s_bound}\\
            \norm{\rstate(t) - \zstate(t)} &\leq \delta(t), \quad t \geq 0,
        \end{align}
    \end{subequations}
    for trajectories $\rstate(t)$, $\zstate(t)$, $\xnerror(t)$ satisfying \eqref{eq:truemodal}, \eqref{eq:nominal}, and \eqref{eq:xndoterror} respectively, and $s(t)$, $\delta(t)$ satisfying
    \begin{subequations}
    \label{eq:s_delta_dynamics}
        \begin{align}
            \dot{s} =& \lambda_{\An} s + (1 + L_{\wnl}) (L_{\fnl} s + \bar{d}) \label{eq:s_dynamics}\\
            &\quad\quad\quad\quad\quad\quad\, + \norm{ \Bn \ctrl} + L_{\wnl} \norm{\Br \ctrl } \nonumber \\
            \dot{\delta} =& (\lambda_{\Ar} + L_{\rnl})\delta + L_{\fnl} s + \bar{d}. \label{eq:delta_dynamics}
        \end{align}
    \end{subequations} 
\end{proposition}

\begin{proof}
First, notice that $\xnerror^\top \An \xnerror = \frac{1}{2}\xnerror^\top (\An+\An^\top) \xnerror \leq \lambda_{\An} \norm{\xnerror}^2$, where $\lambda_{\An}$ is the real part of the slowest eigenvalue of $\An$. Similarly, we have $(\rstate - \zstate)^\top \Ar (\rstate - \zstate) \leq \lambda_{\Ar} \norm{\rstate - \zstate}^2$, where $\lambda_{\Ar}$ is the real part of the largest eigenvalue of $\Ar$.

Assume for simplicity that $\xnerror \neq 0$, then we have that
\begin{align*}
    &\frac{\mathrm{d}}{\mathrm{d}t}\norm{\nstate - \orthogonal^\top \wnl(\rstate)} = \frac{\xnerror^\top \xndoterror}{\norm{\xnerror}} \\
    \stackrel{\eqref{eq:xndoterror}}{=}& \frac{\xnerror^\top }{\norm{\xnerror}} \left( \An \xnerror + \orthogonal^\top \left(\eye  - \wnl'(\rstate) \tangent^\top \right) \left( \mathbf{B} \ctrl + \manError(\fstate)+\dd \right) \right)\\
    \leq& \lambda_{\An} \norm{\xnerror} + \norm{\orthogonal^\top \left(\eye-\wnl'(\rstate) \tangent^\top \right) (\manError(\fstate) +\dd)}  \\
    &+ \norm{\orthogonal^\top \left(\eye  - \wnl'(\rstate) \tangent^\top \right)\mathbf{B} \ctrl}\\
    \leq& \lambda_{\An} \norm{\xnerror} + \norm{\eye-\wnl'(\rstate) \tangent^\top} \norm{\manError(\fstate) + \dd}   \\
    &+ \norm{ \Bn \ctrl} + \norm{\wnl'(\rstate)} \norm{\Br \ctrl }\\
    \stackrel{\eqref{eq:manError}}{\leq}& \lambda_{\An} \norm{\xnerror} + \left(1 + \norm{\wnl'(\rstate)}\right) (L_{\fnl} \norm{\xnerror} + \norm{\dd})  \\
    &+ \norm{ \Bn \ctrl} + \norm{\wnl'(\rstate)} \norm{\Br \ctrl }\\
    \leq&  \lambda_{\An} \norm{\xnerror} + (1 + L_{\wnl}) (L_{\fnl} \norm{\xnerror} + \bar{d})\\
    &+ \norm{ \Bn \ctrl} + L_{\wnl} \norm{\Br \ctrl }
\end{align*}
where the third line is due to Cauchy-Schwarz.
The fourth line uses the fact that $\mathbf{V}$ is orthonormal, so $\norm{\orthogonal^\top} =1$ and $\norm{\tangent^\top} =1$, while the last inequality is due to Assumption~\ref{assum:contdiff}.
Let $s(t)$ be the solution of \eqref{eq:s_dynamics} with $s(0) \geq \norm{\xnerror(0)}$. By the comparison lemma \cite[Lemma 3.4]{khalil2002nonlinear}, 
$$\norm{\xnerror(t)} \leq s(t), \quad t \geq 0.$$

Again, assume for simplicity that $\rstate(t) - \zstate(t)\neq 0$, then, the scalar error dynamics on the manifold can be derived as
\begin{align*}
   &\frac{\mathrm{d}}{\mathrm{d}t}\norm{\rstate - \zstate}
   = \frac{(\rstate - \zstate)^\top (\rdotstate - \zdotstate)}{\norm{\rstate - \zstate}} \\
   \stackrel{\eqref{eq:truemodal}, \, \eqref{eq:nominal}}{=}& \frac{(\rstate - \zstate)^\top }{\norm{\rstate - \zstate}}\left(\Ar (\rstate - \zstate) \right. \\ 
   & \left.+ \rnl(\rstate) - \rnl(\zstate) + \tangent^\top  (\manError(\fstate) + \dd )\right) \\
   \leq& \lambda_{\Ar} \norm{\rstate - \zstate}  + L_{\rnl}\norm{\rstate - \zstate} + \norm{\tangent^\top} \norm{\manError(\fstate) + \dd}\\
   \stackrel{\eqref{eq:manError}}{\leq}&(\lambda_{\Ar} + L_{\rnl})\norm{\rstate - \zstate} + L_{\fnl} \norm{\xnerror} + \norm{\dd}\\
   \stackrel{\eqref{eq:s_bound}}{\leq}&(\lambda_{\Ar} + L_{\rnl})\norm{\rstate - \zstate} + L_{\fnl} s + \bar{d}.
\end{align*}
Similarly, let $\delta(t)$ be a solution of \eqref{eq:delta_dynamics} with $\delta(0) \geq \norm{\rstate(0) - \zstate(0)}$. By the comparison lemma, 
\begin{align*}
    \norm{\rstate(t) - \zstate(t)} &\leq \delta(t), \quad t \geq 0. \qedhere
\end{align*}
\end{proof}
These dynamic uncertainty tubes allow us to design MPC schemes that are less conservative than robust schemes which rely on worse-case analysis, such as rigid-tube MPC.

In the following sections, we will use $s$ and $\delta$ as constraint tightening tubes to ensure constraint satisfaction. The tube dynamics \eqref{eq:s_dynamics} and \eqref{eq:delta_dynamics} can be made stable by constraining the system sufficiently close to the origin until the Lipschitz constants satisfy $\lambda_{\An} + L_{\fnl}(1 + L_{\wnl}) \leq 0$ and $\lambda_{\Ar} + L_{\rnl} \leq 0$. In deriving \eqref{eq:s_dynamics}, we separated the input into its $\Bn \ctrl$ and $\Br \ctrl$ components to exploit the known directional information and get a tighter bound.

\section{Robust MPC}
\label{sec:robustMPC}
In the following, we use the tube dynamics (Prop.~\ref{prop:tubedyn}) to derive a robust MPC formulation. 
To this end, the following proposition shows we can ensure constraint satisfaction by posing more restrictive tightened constraints on the prediction state $\zstate$.
\begin{proposition} \label{prop:tighten}
    Suppose $\norm{\xnerror} \leq s$, $\norm{\rstate - \zstate} \leq \delta$, and 
    \begin{equation} \label{eq:tighten}
    h_j(\C \wmap(\zstate)) + L_{h_j} L_{\C \wmap}\delta + L_{h_j} \norm{\C}s \leq 0
    \end{equation}
    $\forall j = 1, \ldots, n_{h}$, where $L_{\C \wmap}$, $L_{h_j}$ are the Lipschitz constants of $\C \wmap$ and $h_j$, respectively. Then, $\y = \C\fstate$ satisfies the constraints \eqref{eq:constraints}.
\end{proposition}
Using these tightened constraints, we now formulate our proposed RN-ROMPC as follows: 
\begin{subequations} \label{eq:homotheticOCP}
\begin{align}
        &\min_{\ctrl(\cdot), \zstate(\cdot)}~ \int_{0}^{T_{\mathrm{f}}} \ell(\zstate(\tau),\ctrl(\tau))\mathrm{d}\tau +\ell_{\mathrm{f}}(\zstate(T_{\mathrm{f}})) \nonumber \\
    \mathrm{s.t.}~
        & \delta(0) = \norm{\rstate(t) - \zstate(0)},\\
        \label{eq:homotheticOPC_init_s}
        & s(0) = s_0,\\
        &\zdotstate =\rom(\zstate)+\Br \ctrl, \label{eq:zDyn}\\
        &\dot{s} = \lambda_{\An} s + (1 + L_{\wnl})(L_{\fnl} s + \bar{d}) \label{eq:sDyn}\\
            & \quad\quad\quad\quad\quad\quad\, + \norm{ \Bn \ctrl} + L_{\wnl} \norm{\Br \ctrl }, \nonumber \\
        &\dot{\delta} = (\lambda_{\Ar} + L_{\rnl})\delta + L_{\fnl} s + \bar{d}, \label{eq:deltaDyn}\\
        &h_j(\C \wmap(\zstate)) + L_{h_j} L_{\C \wmap}\delta + L_{h_j} \norm{\C}s \leq 0, \label{eq:MPCconstraint}\\
        &\ctrl(\tau) \in \Cntrl \\
        & (\zstate(T_{\mathrm{f}}),\delta(T_{\mathrm{f}}),s(T_{\mathrm{f}}))\in\mathcal{X}_{\mathrm{f}} \label{eq:termConstraint}\\
        & \tau \in [0, T_{\mathrm{f}}], \quad j = 1,\ldots, n_{h}. \nonumber
\end{align}
\end{subequations}
Here, $T_{\mathrm{f}}>0$ represents the prediction horizon. The trajectories of $\delta,\, s,\, \zstate$ are obtained through the tube propagation outlined in Proposition~\ref{prop:tubedyn} and are subject to the tightened constraints in Proposition~\ref{prop:tighten}. The measured reduced-order state $\rstate$ and a variable $s_0 \geq \norm{\xnerror}$ provide the initial conditions. We denote the optimal solution to Problem~\eqref{eq:homotheticOCP} with a star ($^\star$).

The following algorithms summarize the overall design and closed-loop operation.
\begin{algorithm}[H]
\caption{Offline design}
Determine ROM in \eqref{eq:rom} and \eqref{eq:nominal} from known model \cite{jain2021compute} or from data \cite{AloraCenedeseEtAl2023}.\\
Compute Lipschitz constants in \eqref{eq:s_delta_dynamics} and \eqref{eq:tighten}.\\
Design terminal cost/set $\ell_{\mathrm{f}}$, $\mathcal{X}_{\mathrm{f}}$ (Assumption~\ref{assum:terminal}).
\label{alg:offline}
\end{algorithm}
\begin{algorithm}[H]
\caption{Online operation}
At $t=0$: Initialize $s_0\geq \|\xnerror\|$
\begin{algorithmic}
\For{each sampling time $t_k = k \Delta$, $k \in \N$}
\State Measure reduced state $\rstate(t_k)$
\State Solve Problem~\eqref{eq:homotheticOCP}
\State Apply input $\ctrl^\star(\tau)$, $\tau\in[0,\Delta)$
\State Set initial value $s_0=s^\star(\Delta)$.
\EndFor
\end{algorithmic}
\label{alg:online}
\end{algorithm}
In the following, we consider for simplicity a quadratic stage cost $\ell(\zstate,\ctrl):=\|\zstate-\bar{\mathbf{z}}_{\mathrm{r}}\|_Q^2+\|\ctrl-\bar{\ctrl}\|_R^2$ with positive definite matrices $Q,R$ and some nominal steady-state $r(\bar{\mathbf{z}}_{\mathrm{r}})+\Br\bar{\ctrl}=0$. 
To ensure closed-loop guarantees, we also require suitable conditions on the terminal cost $\ell_{\mathrm{f}}$ and the terminal set $\mathcal{X}_{\mathrm{f}}$, as standard in MPC (cf.~\cite{rawlings2017model}).
\begin{assumption} \label{assum:terminal}
There exists
a terminal control law $\kappa:\mathbb{R}^n\rightarrow\mathcal{U}$, 
such that 
for any $(\zstate(0),\delta(0),s(0))\in\mathcal{X}_{\mathrm{f}}$, the trajectories $\zstate(\tau),\delta(\tau),s(\tau)$ according to~\eqref{eq:nominal}, \eqref{eq:s_delta_dynamics} with $\ctrl (\tau) = \kappa(\zstate(t))$ satisfy:
    \begin{enumerate}[i.]%
        \item positive invariance: $(\zstate(\Delta),\delta(\Delta),s(\Delta))\in\mathcal{X}_{\mathrm{f}}$
        \label{assum:terminal_PI} 
        \item constraint satisfaction: $\zstate(\tau),\delta(\tau),s(\tau)$ satisfy \eqref{eq:tighten} for all $\tau\in[0,\Delta]$
        \label{assum:terminal_constraints}
        \item control Lyapunov function: \\$\ell_{\mathrm{f}}(\zstate(\Delta))-\ell_{\mathrm{f}}(\zstate(0)) \leq \int_{0}^{\Delta} \ell(\zstate(\tau),\ctrl(\tau))\mathrm{d}\tau$.
        \label{assum:terminal_CLF}
    \end{enumerate}
    where $\Delta \geq 0$ represents the sampling period. %
\end{assumption}
The simplest way to construct such a terminal set is $\mathcal{X}_{\mathrm{f}}=\{(\zstate,\delta,s)|~\zstate=\bar{\mathbf{z}}_{\mathrm{r}}, (\delta,s)\in\mathcal{A}\}$, $\ell_{\mathrm{f}}=0$, $\kappa=\bar{\ctrl}\in\mathcal{U}$ where $\mathcal{A}\subseteq\mathbb{R}^2_{\geq 0}$ is a positive invariant set for the linear dynamics~\eqref{eq:s_delta_dynamics}, in the linear constraint set~\eqref{eq:tighten}, with constant $\zstate=\bar{\mathbf{z}}_{\mathrm{r}}$.\footnote{%
A corresponding (e.g. polytopic) set always exists, if $\bar{\ctrl}\in\mathcal{U}$, $GC\bar{\mathbf{z}}_{\mathrm{r}}<g$ and $\bar{d},L_{\fnl},L_{\rnl}>0$ are sufficiently small.}
The following theorem summarizes the theoretical properties of the proposed MPC scheme.
\begin{theorem}
Suppose that the initialization at  $t=0$ satisfies $s_0\geq \|\xnerror(0)\|$ and  that Problem \eqref{eq:homotheticOCP} is feasible at time $t=0$. %
Then, Problem \eqref{eq:homotheticOCP} is feasible for all sampling times $t_k$, $k \in \N$, and the closed loop system resulting from Algorithm~\ref{alg:online} satisfies the constraints \eqref{eq:constraints} for all $t \geq 0$. Furthermore, as $\lim _{k \to \infty}$, the nominal trajectory converges to the desired steady-state, \ie
$\zstate^{\star}=\bar{\mathbf{z}}_{\mathrm{r}}$, $\ctrl^{\star}= \bar{\ctrl} $.
\end{theorem}
\begin{proof}
The following proof utilizes standard MPC arguments (cf.~\cite{rawlings2017model}) and the derived bounds in
Propositions~\ref{prop:tubedyn} and \ref{prop:tighten}.\\
\textbf{Part I. Recursive feasibility:}
    Assume Problem~\eqref{eq:homotheticOCP} is feasible at time $t_k$, $k \in \N$, and let $s^\star(\tau)$, $\delta^\star(\tau)$, $\zstate^\star(\tau)$, $\ctrl^\star(\tau)$ for $\tau \in [0, T_{\mathrm{f}}]$ denote its solution. At time $t_{k+1}$, we consider the following shifted candidate solution
    \begin{subequations}
    \label{eq:candidate}
    \begin{align}
        s(0) =& s^\star(\Delta)\\
        \zstate(0) =& \zstate^\star(\Delta)\\
        \ctrl(\tau) =& \begin{cases} \ctrl^\star(\tau+\Delta), & \tau \in [0, T_{\mathrm{f}}-\Delta]\\  \kappa(\zstate(\tau)), & \tau \in [T_{\mathrm{f}} - \Delta, T_{\mathrm{f}}] \end{cases}
    \end{align}
    \end{subequations}
    with trajectories $\zstate(\tau)$, $s(\tau)$, $\tau \in [0, T_{\mathrm{f}}]$ according to the dynamics~\eqref{eq:zDyn}, \eqref{eq:sDyn}. This implies $\zstate(\tau) = \zstate^\star(\tau + \Delta)$, $s(\tau) = s^\star(\tau + \Delta)$ for $\tau \in [0, T_{\mathrm{f}} - \Delta]$. 
  
    By Proposition \ref{prop:tubedyn}, $\delta(0) = \norm{\rstate(0) - \zstate(0)} \leq \delta^\star (\Delta)$. 
    Since the trajectories $\delta(\tau)$ and $\delta^\star(\tau+\Delta)$ are subject to the same (continuous) dynamics~\eqref{eq:deltaDyn} with the same $s$, it follows from the comparison lemma \cite[Lemma 3.4]{khalil2002nonlinear} that $\delta(\tau) \leq \delta^\star(\tau+\Delta)$ for $\tau \in [0, T_{\mathrm{f}} - \Delta]$. 
    Thus, for $\tau \in [0, T_{\mathrm{f}}-\Delta]$ the candidate solution satisfies the constraints~\eqref{eq:MPCconstraint} with
    \begin{align*}
        &h_j(\C \wmap(\zstate(\tau))) + L_{h_j} L_{\C \wmap}\delta(\tau) + L_{h_j} \norm{\C}s(\tau) \\
        \leq& h_j(\C \wmap(\zstate(\tau+\Delta)))\\
        &+ L_{h_j} L_{\C \wmap}\delta(\tau+\Delta) + L_{h_j} \norm{\C}s(\tau+\Delta)        \stackrel{\eqref{eq:MPCconstraint}}{\leq}0.
    \end{align*}
    By Assumption \ref{assum:terminal}(\ref{assum:terminal_constraints}), we also have that \eqref{eq:MPCconstraint} holds for $\tau \in [T_{\mathrm{f}}-\Delta, T_{\mathrm{f}}]$, i.e., the constraints \eqref{eq:MPCconstraint} hold for $\tau \in [0, T_{\mathrm{f}}]$. 
    Lastly, by Assumption \ref{assum:terminal}(\ref{assum:terminal_PI}), we also have that \eqref{eq:termConstraint} holds. Thus, the MPC resulting from Algorithm~\ref{alg:online} is recursively feasible.\\
    \textbf{Part II. Constraint satisfaction:} 
    First, due to the fixed initial condition of $s$ in Algorithm~\ref{alg:online} and~\eqref{eq:homotheticOPC_init_s}, $s_0$ satisfies the dynamics~\eqref{eq:s_dynamics} also across optimization steps. 
     Hence, the initialization $s_0\geq \|\xnerror(0)\|$ and Proposition~\ref{prop:tubedyn} ensures that at each sampling time $t_k$:  $s_0\geq \|\xnerror(t_k)\|$ holds recursively. 
      Furthermore, applying Proposition~\ref{prop:tubedyn} in the interval $\tau\in[0,\Delta)$ yields
      $\|\rstate(t_k+\tau)-\zstate^\star(\tau)\|\leq \delta^\star(\tau)$, $\|\xnerror(\tau+t_k)\|\leq s(\tau)$. 
      Finally, Proposition~\ref{prop:tighten} and the tightened constraints~\eqref{eq:MPCconstraint} for $\tau\in[0,\Delta)$ yield 
      $h_j(\C\fstate(t))\leq 0$, $\forall j=1,\ldots,n_h$, $t\in[t_k,t_{k+1})$, i.e., the constraints~\eqref{eq:constraints} hold for all $t\geq 0$. \\
 \textbf{Part III. Convergence/stability:}
 The cost function in~\eqref{eq:homotheticOCP}, candidate solution~\eqref{eq:candidate}, and terminal cost condition in Assumption~\ref{assum:terminal}(\ref{assum:terminal_CLF}) are equivalent to nominal MPC with state $\zstate$~\cite{rawlings2017model}. 
Hence, following standard arguments, it holds that
\begin{align*}
\sum_{k=0}^\infty\int_0^\tau\ell(\zstate(\tau),\ctrl(\tau))<\infty,
\end{align*}
and Barbalat's Lemma~\cite{khalil2002nonlinear} ensures convergence, see, \eg~\cite[Thm.~12]{sasfi2022robust} for details.
\end{proof}
As is common in robust MPC, a linear tube-feedback $\ctrl = \mathbf{K} \rstate + \mathbf{c}$ can be used to reduce conservatism \cite{kouvaritakis2016model}, but in this work, we solely focus on open-loop prediction for simplicity.

\section{Discussion}
\label{sec:discussion}
In the following, we discuss the qualitative properties of the proposed RN-ROMPC scheme and practical implementation aspects for data-driven models.
\subsection{Properties of Robust RN-ROMPC Scheme}
We now discuss several properties of our proposed robust RN-ROMPC scheme. First, for a full order system $\nrstate = \nfstate$, the proposed RN-ROMPC scheme is comparable to a robust MPC scheme using a homothetic tube, where $\delta\geq 0$ is the corresponding scaling (cf.~\cite{sasfi2022robust,rakovic2022homothetic}). The difference is that we account for errors due to the SSM-based reduction scheme by exploiting the invariance properties of the manifold. These properties allow us to decompose the error dynamics into an off-manifold error component $\xnerror$ and an on-manifold component $\rstate-\zstate$.

Second, for small enough Lipschitz constants $L_{\fnl}$, $L_{\wnl}$, and $L_{\rnl}$, \ie $|(1 + L_{\wnl})L_{\fnl}| < |\lambda_{\An}|$ and $|L_{\rnl}| < |\lambda_{\Ar}|$, the tube dynamics are stable. 
At least locally, we expect the off-manifold error dynamics to be stable due to the expected time-scale separation between $\lambda_{\Ar}$ and $\lambda_{\An}$, \ie $|\lambda_{\Ar}| \gg |\lambda_{\An}|$ along with the fact that the Lipschitz constants are arbitrarily small for analytic functions (in a small enough neighborhood of the origin). To reduce conservativeness in the $\delta$ predictions, we could use an additional linear feedback $\ctrl = \mathbf{K} \rstate + \mathbf{c}$ to ensure that $|\lambda_{\Ar}| \gg L_{\rnl}$.

Lastly, according to the tube dynamics \eqref{eq:s_delta_dynamics}, as $\norm{\ctrl(t)}$ increases, the orthogonal error increases. The addition of control input leads to the excitement of the fast modes; thus, the SSM is no longer invariant. Applying large inputs results in large values of $s$ and, in turn, large on-manifold error, $\delta$. This causes the model's uncertainty to grow, increasing the constraint tightening in \eqref{eq:tighten}.
Hence, if we wish to operate close to the constraints, the proposed MPC policy will implicitly act cautiously to reduce the excitation of the fast modes.

\begin{figure}
    \centering
    \includegraphics[width=0.47\textwidth]{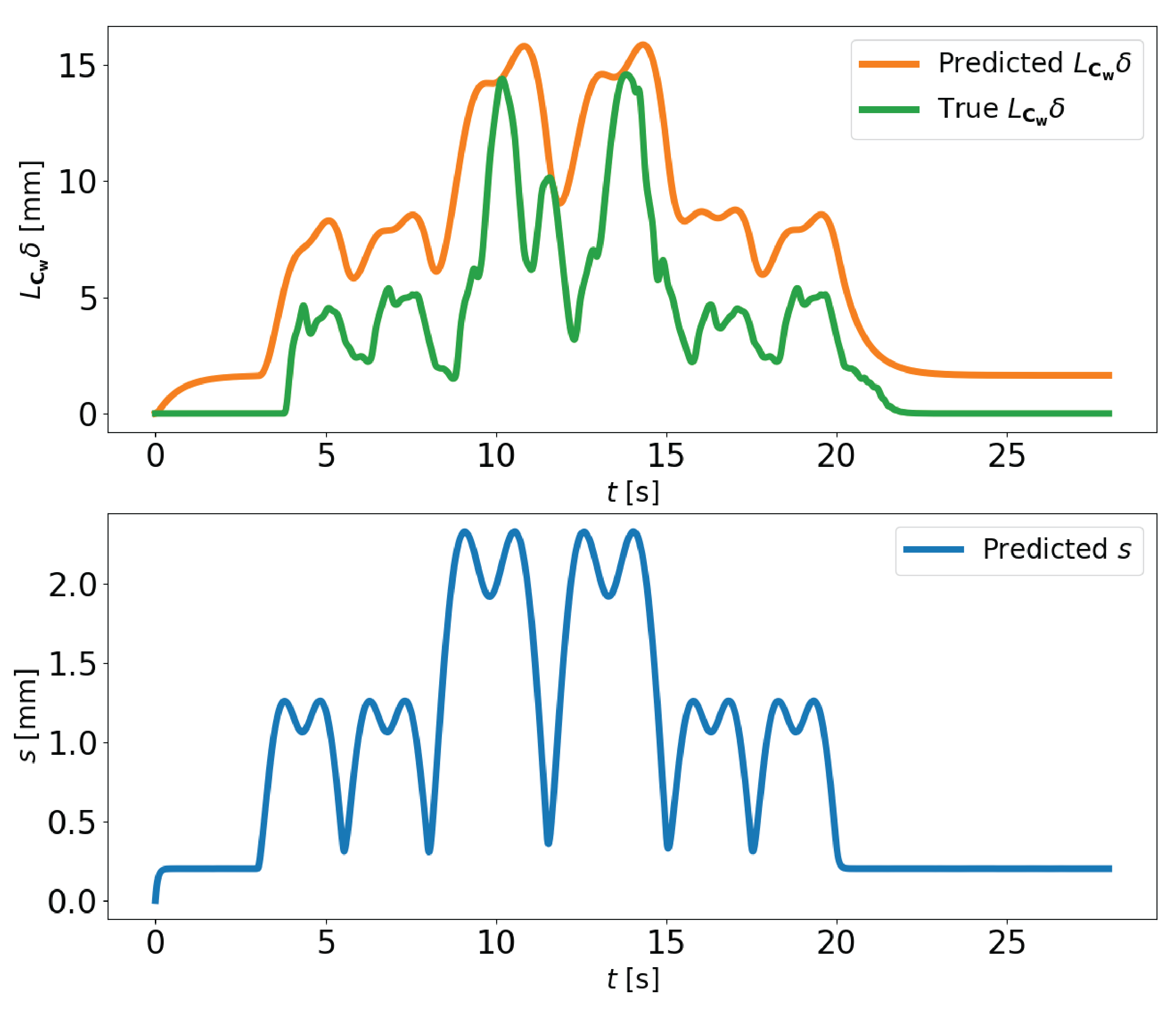}
    \caption{Trajectory of fitted dynamics $\delta(t)$ in \eqref{eq:delta_dynamics} (top) and $s(t)$ in \eqref{eq:new_s_dynamics} (bottom). At zero input, the predicted $s$ and $\delta$ settle to non-zero steady state due to the bounding disturbances $\bar{d}, \hat{d} > 0$. We only plot predicted $s$ because, in practice, we do not have access to the fast modes, $\nstate$.}
    \label{fig:fittedTubeDyn}
    \vspace{-0.65cm}
\end{figure}
\subsection{Data-Driven Reduced-Order Model}
\label{subsec:dataDriven}
In the previous sections, we extract ROMs directly from a known model $\mathbf{f}$. This may be difficult even in a simulation environment since extracting the full-order model from finite element code is a cumbersome and code-intrusive process. Furthermore, for real-world experiments, we may want to extract reduced models directly from observation data.
\begin{figure*}
\centering
\includegraphics[width=0.92\textwidth]{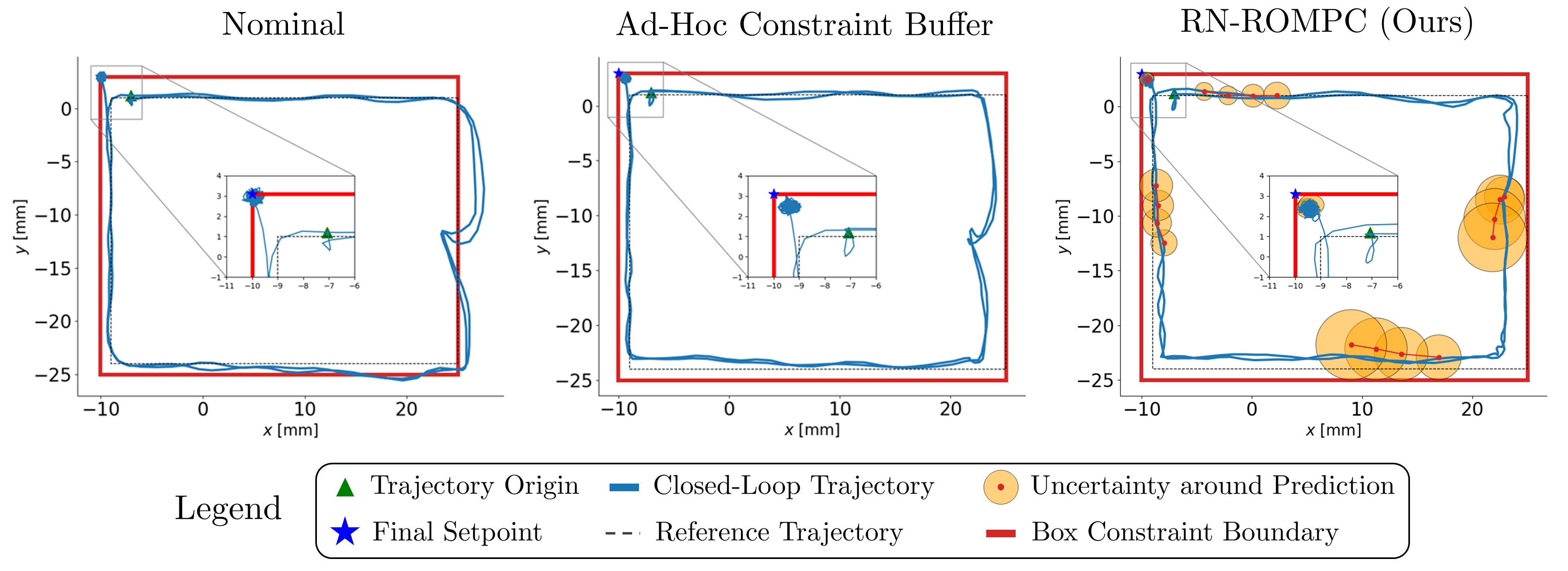}
  \caption{Simulation results comparing SSM-based MPC with soft constraints (left), an ad-hoc constraint buffer scheme (center), and our proposed RN-ROMPC scheme (right) for a square periodic trajectory and setpoint. The trajectory has a period of 1 second and the prediction horizon for the MPC schemes is $T_f = 0.06$ seconds.}
  \label{fig:squareTraj}
  \vspace{-0.6cm}
\end{figure*}

To alleviate these challenges, we use the data-driven approach described in \cite{cenedese2022ssmlearn} to extract ROMs on SSMs. 
With this approach, we can estimate the reduced-order state $\rstate$ online using past output measurements $\mathbf{y}$.
We refer the interested reader to \cite{AloraCenedeseEtAl2023} for more details. In this case, $\bar{d}$ in \eqref{eq:s_dynamics} does not only account for the disturbances $\mathbf{d}(t)$ but also for induced regression and truncation error in the estimation of the SSM from limited and noisy data. Furthermore, we conduct a coordinate transformation such that $\Ar$ is in real block-diagonal form whose entries are in decreasing order of the real parts of its eigenvalues. Note that since we do not observe the fast modes, we never need to diagonalize System~\eqref{eq:FOM}.

Since we construct ROMs directly from output data, we do not have access to $\B$, $\orthogonal$, and $\wmap$ as required in \eqref{eq:s_dynamics}. To overcome this practical challenge, one can derive the following alternative off-manifold scalar bounding error dynamics:
\begin{equation}
\label{eq:new_s_dynamics}
    \dot{s} = (\lambda_{\An} + \bar{L}) s + \bar{B} \norm{\ctrl} + \hat{d},
\end{equation}
with the new constants $\bar{L} \geq \norm{\left(\eye  - \wnl'(\rstate) \tangent^\top \right)}L_{\fnl}$, $\bar{B} \geq \norm{\left(\eye  - \wnl'(\rstate) \tangent^\top \right)\mathbf{B}}$, and $\hat{d} \geq (1 + L_{\wnl})\bar{d}$ (see Proposition~\ref{prop:tubedyn}). 

Since we do not know these constants, we fit the tube dynamics of $\delta$ in \eqref{eq:delta_dynamics} and the new dynamics of $s$ in \eqref{eq:new_s_dynamics} from data by applying a sequence of open-loop control inputs $\ctrl(t)$ and estimating $\rstate(t)$ where $t \in [0, t_f]$. We then integrate the reduced dynamics in \eqref{eq:nominal} using the control inputs $\ctrl(t)$ to get the sequence $\mathbf{z}_r(t)$. The constants are fitted by solving the following optimization problem for a fixed $\hat{d}$ and $\bar{d}$.
\begin{equation}
\begin{aligned}
        \minimize_{L_{\fnl}, L_{\rnl}, \bar{B}}~
        & \int_{0}^{t_f} (\norm{\Gj} L_{\C \wmap} \delta + \norm{\Gj \C} s)dt \\
    \mathrm{subject~to}~
        & \text{System~\eqref{eq:delta_dynamics}}, \\
        & \text{System~\eqref{eq:new_s_dynamics}}, \\
        & \delta(t) \geq \norm{\rstate(t) - \mathbf{z}_r(t)}, \\
        & \delta(0) = s(0) = 0, \quad L_{\fnl}, L_{\rnl}, \bar{B} \geq 0,
\end{aligned}
\end{equation}
where $\mathbf{G}$ and $\mathbf{g}$ are the matrix and vector representing polytopic constraints, respectively, and $\mathbf{G}_j$ represents the $j$-th row such that the $j$-th constraint is $h_j(\obs) = \mathbf{G}_j \obs - \mathbf{g}_j$. This optimization problem solves for the appropriate constants that minimize the constraint tightening and are consistent with the generated data.

The fitting data is generated by applying zero inputs, followed by a sequence of alternating moderate and large inputs, and at last, zero inputs again for five seconds each. The data is generated with noisy inputs with $2$-norm of $400$ Newtons sampled from a Gaussian distribution. Figure~\ref{fig:fittedTubeDyn} shows the fitted tube dynamics under an open-loop control sequence $\ctrl(t)$. For $\bar{d} = \hat{d} = 3$ mm, the optimized constants are $L_{\fnl} = 120.897$, $L_{\rnl} = 2.019$, $\bar{L} = 0.001$, and $\bar{B} = 0.012$. Note that the tube dynamics are stable, and our upper bound closely tracks the true error of the system.

\section{SIMULATION RESULTS}
\label{sec:results}
In this section, we highlight the robustness properties of the proposed RN-ROMPC scheme in simulation.

\subsection{Setup}
We consider the control of an elastomer ``Diamond" soft robot. We conduct simulations using the SOFA framework based on the finite element method \cite{allard2007sofa}. The robot mesh used for simulation is available in the \textit{SoftRobots} plugin \cite{coevoet2017software}, and the parameters of the Diamond robot match those described in \cite{AloraCenedeseEtAl2023}: the Diamond robot has a mass of $0.45$ kg, Poisson ratio of 0.45, and Young's modulus of $175$ MPa. The finite element model has $1628$ nodes, leading to a $\nfstate = 9768$ dimensional state space. The damping is modeled with Rayleigh (proportional) damping.

We implement the proposed MPC scheme in the open-source soft robot control library\footnote[1]{https://github.com/StanfordASL/soft-robot-control} and learn a 6-dimensional ROM of the Diamond robot using the Spectral Submanifold Reduction for control library\footnote[2]{https://github.com/StanfordASL/SSMR-for-control} according to the procedure in \cite{AloraCenedeseEtAl2023}. We consider a receding horizon of $3$ time steps with a control sampling time of $dt = 0.02$ seconds. Problem~\eqref{eq:homotheticOCP}  is solved using sequential convex programming \cite{BonalliCauligiEtAl2019}.

\subsection{Results}
To demonstrate the efficacy of our proposed robust RN-ROMPC scheme, we consider a trajectory tracking problem where the robot tip is meant to follow a reference trajectory. The reference first corresponds to a periodic square reference with a 1-second period, which touches the constraints on the right and then converges to a setpoint on the top left corner of the constraints. Figure~\ref{fig:squareTraj} depicts a comparison of our proposed approach against a nominal MPC scheme (left) and an ad-hoc constraint buffer MPC scheme (right). 
The ad-hoc buffer is implemented by artificially tightening the original constraint bounds and the tightening is chosen to minimize conservativeness while remaining within the constraints. 
In both the nominal and ad-hoc buffer schemes, the constraints on $\mathbf{y}$ are treated as soft constraints. Additionally, we consider the noise $\dd(t) = \mathbf{B}\ctrl_d(t)$ where $\norm{\ctrl_d} = 400$ and $0 \leq u_i \leq 2500$, $i = 1,\dots,\nctrl$.

The nominal MPC  significantly violates the right border constraint and leaves the constraint set when attempting to track the setpoint. This is due to the fact that as the robot moves closer to the right border and further from its equilibrium point, the accuracy of the SSM ROM deteriorates, leading to inaccurate prediction in the MPC scheme. To account for this, we considered an ad-hoc constraint buffer scheme where we tightened the right constraint by 3 mm, the bottom by 2 mm, the left by 0.6 mm,  and the top constraint by 0.5 mm to ensure constraint satisfaction. On the other hand, our approach also renders the system safe during its entire operation without any ad-hoc tuning.

A qualitative comparison of the ad-hoc scheme and RN-ROMPC in Figure~\ref{fig:squareTraj} reveals that our approach is not much more conservative. Furthermore, we found that further tightening or loosening of the constraints resulted in more conservative behavior (compared to RN-ROMPC) or constraint violations, respectively. In contrast, our approach maintains the flexibility of being able to tighten the constraints dynamically and thus, handle arbitrary trajectories.

Note that in Figure~\ref{fig:squareTraj}, the uncertainty tubes surrounding the predictions vary in size depending on the robot's position from its equilibrium point. In particular, larger inputs are required as the robot moves further away from its fixed equilibrium point. Thus, as the robot moves towards (and away from) the bottom right corner, the required inputs are largest, resulting in the largest uncertainty tubes. In contrast, the tubes are smaller in parts of the workspace closer to the equilibrium point, \eg the top left corner. This is expected since our tube dynamics depend directly on the magnitude of the control inputs (see~Equation~\eqref{eq:s_delta_dynamics}). Since the uncertainty tubes shrink as the robot moves towards the origin, the controller becomes less conservative and gets nearer to the constraint to more closely track the desired trajectory.
\section{CONCLUSION}
In this work, we considered the problem of robust online optimal control for high-dimensional systems. We derived error bounds on our prediction model using properties of SSMs, formulated a novel robust MPC scheme based on these error bounds, proved that our scheme robustly satisfies constraints, and demonstrated the efficacy of our approach on a challenging, high-dimensional soft robot example in finite element simulation.

\iftoggle{ext}{
\section{Appendix}
\label{app:Appendix}
\begin{proof}[Proof of Lemma~\ref{lem:redAutDyn}]
    Using the definition of $\rstate$ and taking derivatives we have that
    \begin{align*}
    \rdotstate\stackrel{\eqref{eq:rom}}{=}&\tangent^\top \dot{\fstate}\stackrel{\eqref{eq:FOM_aut}}{=} 
    \tangent^\top (\mathbf{A} \fstate + \fnl(\fstate))\\
    \stackrel{\eqref{eq:rom}}{=}&\tangent^\top (\mathbf{A} (\tangent \rstate + \wnl(\rstate)) + \fnl(\wmap(\rstate))\\
    =&\Ar \rstate + \rnl(\rstate),
    \end{align*} 
where $\Ar \defeq \tangent^\top \mathbf{A}\tangent$ and $\rnl \defeq \tangent^\top \fnl(\wmap(\rstate))$. The last equality follows from applying invertibility \eqref{eq:invert} to the definitions of $\vmap$ and $\wmap$, \ie 
    \begin{align*}
        0=&\vmap(\wmap(\rstate))-\rstate \\
         \stackrel{\eqref{eq:rom}}{=}& \tangent^\top \tangent \rstate + \tangent^\top  \wnl(\rstate)-\rstate
        \stackrel{\eqref{eq:orthbasis}}{=} \tangent^\top \wnl(\rstate),
    \end{align*}
and using this to arrive at $0 = \Ar \tangent^\top \wnl \stackrel{\eqref{eq:orthbasis}}{=} \tangent^\top \mathbf{A} (\eye - \orthogonal\orthogonal^\top)\wnl \stackrel{\eqref{eq:Ar}}{=}\tangent^\top \mathbf{A}\wnl$.
\end{proof}
\begin{proof}[Proof of Lemma~\ref{lem:SSMproperties}]
    To show \eqref{eq:VnVnWnl}, we write
    \begin{align*}
        \orthogonal \orthogonal^\top \wnl(\rstate) \stackrel{\eqref{eq:orthbasis}}{=}& \left( \eye - \tangent \tangent^\top \right) \wnl(\rstate)\\
        =& \, \wnl(\rstate),
    \end{align*}
    where the last equality is due to $\tangent^\top\wnl(\rstate) = \mathbf{0}$.

    Substituting the definition of $\wmap$ \eqref{eq:rom} into the invariance equation \eqref{eq:invar}, we have that
    \begin{align*}
        &\mathbf{A} (\tangent \rstate + \wnl(\rstate)) + \fnl(\wmap(\rstate)) \\
        =& \left( \tangent + \wnl'(\rstate) \right) \rom(\rstate) \\
        \stackrel{\eqref{eq:rautdyn}}{=}& \tangent (\Ar \rstate + \rnl(\rstate)) + \wnl'(\rstate) \rom(\rstate).
    \end{align*}
    Equating nonlinear terms yields
    \begin{equation} \label{eq:nonlin_terms}
        \mathbf{A} \wnl(\rstate) + \fnl(\wmap(\rstate))
        = \tangent \rnl(\rstate) + \wnl'(\rstate) \rom(\rstate). 
    \end{equation}
    To arrive at \eqref{eq:VnDWnl}, we multiply both sides of \eqref{eq:nonlin_terms} by $\orthogonal^\top$ and get
    \begin{align*}
        &\orthogonal^\top \mathbf{A} \wnl(\rstate) + \orthogonal^\top \fnl(\wmap(\rstate))\\
        =&\orthogonal^\top \tangent \rnl(\rstate) + \orthogonal^\top \wnl'(\rstate) \rom(\rstate)\\
        \stackrel{\eqref{eq:orthbasis}}{=}& \orthogonal^\top \wnl'(\rstate) \rom(\rstate).  \qedhere
    \end{align*}
\end{proof}
\begin{proof}[Proof of Lemma~\ref{lem:trueDynModal}]
    Taking the derivative of the reduced component yields
    \begin{align*}
        \rdotstate =& \tangent^\top \dot{\fstate} \\
        \stackrel{\eqref{eq:FOM}}{=}& \tangent^\top \left( \mathbf{A} \fstate + \fnl(\fstate) + \mathbf{B} \ctrl + \dd \right)  \\
        =&\tangent^\top \mathbf{A} \tangent \rstate + \tangent^\top \mathbf{A} \orthogonal \nstate + \tangent^\top \left( \fnl(\fstate) + \mathbf{B} \ctrl + \dd \right) \\
        \stackrel{\eqref{eq:Ar}}{=}&\Ar \rstate + \tangent^\top \left( \fnl (\fstate) + \mathbf{B} \ctrl + \dd \right)\\
        =& \Ar \rstate + \tangent^\top \fnl(\wmap(\rstate)) + \tangent^\top \left(\manError(\fstate) + \B \ctrl + \dd \right)\\
        \stackrel{\eqref{eq:rautdyn}}{=} & \rom(\rstate) + \tangent^\top \left(\manError(\fstate) + \B \ctrl +  \dd \right),
    \end{align*}
    and similarly, for the normal component, we have that
    \begin{align*}
    \ndotstate =& \orthogonal^\top \dot{\fstate} \\
    \stackrel{\eqref{eq:FOM}}{=}& \orthogonal^\top \left( \mathbf{A} \fstate +\fnl(\fstate)+\mathbf{B} \ctrl + \dd \right)\\
    =& \orthogonal^\top \mathbf{A}  \tangent \rstate + \orthogonal^\top \mathbf{A} \orthogonal \nstate
    + \orthogonal^\top \left(\fnl(\fstate)+\mathbf{B} \ctrl + \dd \right)\\
    \stackrel{\eqref{eq:Ar}}{=}&  \An \nstate
    + \orthogonal^\top \left(\fnl(\fstate)+\mathbf{B} \ctrl + \dd \right). \qedhere
    \end{align*}
\end{proof}
\begin{proof}[Proof of Lemma~\ref{lem:manError}]
    The error between the true state and its projection onto the SSM, $\mathcal{W}(E)$ is
    \begin{align} \label{eq:diffxandw}
            \norm{\fstate - \wmap(\rstate)} \stackrel{\eqref{eq:rom}}{=}& \norm{\tangent \rstate + \orthogonal \nstate - \tangent \rstate - \wnl(\rstate)} \nonumber\\
            \stackrel{\eqref{eq:VnVnWnl}}{=}& \norm{\orthogonal \left( \nstate - \orthogonal^\top \wnl(\rstate) \right)} \nonumber\\
            =& \norm{\nstate - \orthogonal^\top \wnl(\rstate)}
    \end{align}
    where the last equality uses the isometry property of $\orthogonal$.

    Thus, using the above and the fact that $\fnl$ is $L_{\fnl}$-Lipschitz, we have that
    \begin{align*}
        \norm{\manError(\fstate)} =& \norm{\fnl(\fstate) - \fnl(\wmap(\rstate))} \\
        \stackrel{\eqref{eq:diffxandw}}{\leq}& L_{\fnl} \norm{\nstate - \orthogonal^\top \wnl(\rstate)}. \qedhere
    \end{align*}
\end{proof}
\begin{proof}[Proof of Lemma~\ref{lem:orthError}]
Taking the derivative of the orthogonal component of the full state on the manifold yields
\begin{equation}
    \begin{aligned} \label{eq:deriv_VnTwnl}
        & \frac{\mathrm{d}}{\mathrm{d}t} \left( \orthogonal^\top \wnl(\rstate) \right) \\
        \stackrel{\text{\eqref{eq:truemodal}}}{=}& \orthogonal^\top  \wnl'(\rstate) \left( \rom(\rstate) + \tangent^\top ( \B \ctrl + \manError(\fstate) + \dd) \right)\\
        \stackrel{\eqref{eq:VnDWnl}}{=}& \orthogonal^\top \left( \mathbf{A} \wnl(\rstate) + \fnl(\wmap(\rstate)) \right)\\
        &+ \orthogonal^\top \wnl'(\rstate) \tangent^\top \left( \B \ctrl +  \manError(\fstate) + \dd \right)\\
        \stackrel{\eqref{eq:VnVnWnl}}{=}&  \An \orthogonal^\top \wnl(\rstate) + \orthogonal^\top \fnl(\wmap(\rstate)) \\
        &+ \orthogonal^\top \wnl'(\rstate) \tangent^\top \left( \mathbf{B} \ctrl + \manError(\fstate) + \dd \right).
    \end{aligned}
\end{equation}
Recalling that $\manError(\fstate) \defeq \fnl(\fstate) - \fnl(\wmap(\rstate))$, we get
\begin{align*}
    \xndoterror=& \ndotstate - \frac{\mathrm{d}}{\mathrm{d}t} \left( \orthogonal^\top \wnl(\rstate) \right)\\
    \stackrel{\eqref{eq:truemodal}, \, \eqref{eq:deriv_VnTwnl}}{=}& \An \xnerror + \orthogonal^\top \left( \fnl(\fstate) + \B \ctrl +  \dd \right) - \orthogonal^\top \fnl(\wmap(\rstate)) \\
    & - \orthogonal^\top \wnl'(\rstate) \tangent^\top \left( \mathbf{B} \ctrl + \manError(\fstate) + \dd \right)\\
    =& \An \xnerror + \orthogonal^\top \left(\eye  - \wnl'(\rstate) \tangent^\top \right) \left( \mathbf{B} \ctrl + \manError(\fstate) + \dd \right). \qedhere
\end{align*}
\end{proof}
\begin{proof}[Proposition~\ref{prop:tighten}]
For any $j = 1, \ldots, n_{h}$, we have
    \begin{align*}
        h_j(\C \fstate) =& h_j(\C \wmap(\rstate)) + h_j(\C \fstate) -h_j(\C \wmap(\rstate))\\
        \leq& h_j(\C \wmap(\zstate)) + \norm{h_j(\C \wmap(\rstate)) - h_j(\C \wmap(\zstate))}\\ &+ \norm{h_j(\C \fstate) -h_j(\C \wmap(\rstate))}\\
        \stackrel{\eqref{eq:x_n_tilde}}{\leq}& h_j(\C \wmap(\zstate)) + L_{h_j}\norm{\C \wmap(\rstate) -\C \wmap(\zstate)}\\ 
        &+ L_{h_j} \norm{\C}\norm{\xnerror}\\ 
        \leq& h_j(\C \wmap(\zstate)) + L_{h_j} L_{\C \wmap}\delta + L_{h_j} \norm{\C}s\\
        \stackrel{\eqref{eq:tighten}}{\leq}&0. \qedhere
    \end{align*}
\end{proof}
}{}

\printbibliography
\end{document}